\documentclass[a4paper,UKenglish,cleveref, autoref, thm-restate]{lipics-v2021}

\title{Positive and monotone fragments of FO and LTL} 

\author{Denis Kuperberg}{CNRS, LIP, ENS Lyon, France \and \url{http://perso.ens-lyon.fr/denis.kuperberg} }{denis.kuperberg@ens-lyon.fr}{https://orcid.org/0000-0001-5406-717X}{ANR ReCiProg}

\author{Quentin Moreau}{ENS Lyon, France}{quentin.moreau@ens-lyon.fr}{}{}

\authorrunning{D. Kuperberg and Q. Moreau} 

\Copyright{Denis Kuperberg and Quentin Moreau} 



\category{} 

\nolinenumbers 

\EventEditors{John Q. Open and Joan R. Access}
\EventNoEds{2}
\EventLongTitle{42nd Conference on Very Important Topics (CVIT 2016)}
\EventShortTitle{CVIT 2016}
\EventAcronym{CVIT}
\EventYear{2016}
\EventDate{December 24--27, 2016}
\EventLocation{Little Whinging, United Kingdom}
\EventLogo{}
\SeriesVolume{42}
\ArticleNo{23}

\usepackage{graphicx}
\usepackage{color}
\usepackage{comment}

\usepackage{amssymb,amsthm,amsmath}
\usepackage[ruled,lined]{algorithm2e}
\usepackage{hyperref}
\usepackage{cleveref}

\usepackage{stmaryrd} 

\renewcommand{\alph}{\part(\Sigma)}

\newcommand{\Tau}{\mathrm{T}}

\newcommand{\N}{\mathbb{N}}

\newcommand{\op}{\cdot}

\newcommand{\dualmon}{\curlywedge}

\newcommand{\FO}{\mathrm{FO}}
\newcommand{\FOp}{\FO^+}
\newcommand{\FOth}{\FO^3}
\newcommand{\FOthp}{\FO^{3+}}
\newcommand{\FOtw}{\FO^2}
\newcommand{\FOtwp}{\FO^{2+}}

\newcommand{\LTL}{\mathrm{LTL}}
\newcommand{\LTLp}{\LTL^+}
\newcommand{\TL}{\mathrm{TL}}
\newcommand{\TLp}{\TL^+}
\newcommand{\UTL}{\mathrm{UTL}}
\newcommand{\UTLp}{\UTL^+}

\newcommand{\EF}{\mathrm{EF}}

\newcommand{\X}{\mathrm{X}}
\newcommand{\Y}{\mathrm{Y}}
\newcommand{\U}{\mathrm{U}}
\renewcommand{\S}{\mathrm{S}}
\newcommand{\Q}{\mathrm{Q}}
\newcommand{\XU}{\mathrm{XU}}
\newcommand{\YS}{\mathrm{YS}}
\newcommand{\R}{\mathrm{R}}

\renewcommand{\P}{\mathrm{P}}
\newcommand{\F}{\mathrm{F}}
\renewcommand{\H}{\mathrm{H}}
\newcommand{\G}{\mathrm{G}}

\newcommand{\toFO}[1]{#1^\bigstar}

\newcommand{\qr}{\mathrm{qr}}

\renewcommand{\part}{\mathcal{P}}
\renewcommand{\succ}{\mathrm{succ}}
\renewcommand{\nsucc}{\mathrm{nsucc}}
\newcommand{\be}{\mathfrak{be}}

\newcommand{\val}{\nu}
\newcommand{\bin}{\mathfrak{B}}
\renewcommand{\b}{\mathfrak{b}}

\newcommand{\M}{\mathbf{M}}

\renewcommand{\L}{\textsc{LogSpace}}
\newcommand{\NL}{\mathsf{NL}}
\newcommand{\comp}{\mathsf{complete}}
\newcommand{\PSPACE}{\textsc{PSpace}}

\newcommand{\leqml}{\leq_{\mathbf{M}_L}}
\newcommand{\geqml}{\geq_{\mathbf{M}_L}}

\newcommand{\leqa}{\leq_{\alph^*}}
\newcommand{\geqa}{\geq_{\alph^*}}

\newcommand{\aut}{\mathcal{B}}

\begin{document}

\maketitle

\begin{abstract}
    We study the positive logic $\FOp$ on finite words, and its fragments, pursuing and refining the work initiated in \cite{PFO}.
    First, we transpose notorious logic equivalences into positive first-order logic: $\FOp$ is equivalent to $\LTLp$, and its two-variable fragment $\FOtwp$ with (resp. without) successor available is equivalent to $\UTLp$ with (resp. without) the ``next'' operator $\X$ available.  This shows that despite previous negative results, the class of $\FOp$-definable languages exhibits some form of robustness.
    We then exhibit an example of an $\FO$-definable monotone language on one predicate, that is not $\FOp$-definable, refining the example from \cite{PFO} with $3$ predicates. 
    Moreover, we show that such a counter-example cannot be $\FOtw$-definable. 

\end{abstract}


\section{Introduction}

In various contexts, monotonicity properties play a pivotal role. For instance the field of monotone complexity investigates negation-free formalisms, and turned out to be an important tool for complexity in general \cite{GrigniSipser92}.
From a logical point of view, a sentence is called monotone (with respect to a predicate $P$) if increasing the set of values where $P$ is true in a structure cannot make the evaluation of the formula switch from true to false. This is crucial e.g. when defining logics with fixed points, where the fixed points binders $\mu X$ can only be applied to formulas that are monotone in $X$. Logics with fixed points are used in various contexts, e.g. to characterise the class \textsc{PTime} on ordered structures \cite{Immerman,Vardi}, as extensions of linear logic such as $\mu$MALL \cite{muMALL}, or in the $\mu$-calculus formalism used in automata theory and model-checking \cite{mucalc}. Because of the monotonocity constraint, it is necessary to recognise monotone formulas, and understand whether a syntactic restriction to positive (i.e. negation-free) formulas is semantically complete.
Logics on words have also been generalised to inherently negation-free frameworks, such as in the framework of cost functions \cite{CostFun}.

This motivates the study of whether the semantic monotone constraint can be captured by a syntactic one, namely the removing of negations, yielding the class of positive formulas.
 For instance, the formula $\exists x, a(x)$ states that an element labelled $a$ is present in the structure. It is both monotone and positive. However, its negation $\forall x, \neg a(x)$  is neither positive nor monotone, since it states the absence of $a$, and increasing the domain where predicate $a$ is true in a given structure could make the formula become false.

Lyndon's preservation theorem \cite{Lyndon59} states that on arbitrary structures, every monotone formula of First-Order Logic ($\FO$) is equivalent to a positive one ($\FOp$ syntactic fragment).
The case of finite structures was open for two decades until Ajtai and Gurevich \cite{AjtaiGurevich87} showed that Lyndon's theorem does not hold in the finite, later refined by StolBoushkin \cite{Stol95} with a simpler proof. Recently, this preservation property of $\FO$ was more specifically shown to fail already on finite graphs and on finite words by Kuperberg \cite{PFO}, implying the failure on finite structure with a more elementary proof than \cite{AjtaiGurevich87,Stol95}.
However, the relationship between monotone and positive formulas is still far from being understood.
On finite words in particular, the positive fragment $\FOp$ was shown \cite{PFO} to have undecidable membership (with input an $\FO$ formula, or a regular language), which could be interpreted as a sign that this class is not well-behaved.
This line of research can be placed in the larger framework of the study of preservation theorems in first-order logic, and their behaviour in the case of finite models, see \cite{Rossman08} for a survey on preservation theorems.

In this work we will concentrate on finite words, and investigate this ``semantic versus syntactic'' relationship for fragments of $\FO$ and Linear Temporal Logic ($\LTL$). We will in particular lift the classical equivalence between $\FO$ and $\LTL$ \cite{Kamp} to their positive fragments, showing that some of the robustness aspects of $\FO$ are preserved in the positive fragment, despite the negative results from \cite{PFO}. This equivalence between $\FO$ and $\LTL$ is particularly useful when considering implementations and real-world applications, as $\LTL$ satisfiability is $\PSPACE$-complete while $\FO$ satisfiability is non-elementary.
It is natural to consider contexts where specifications in LTL can talk about e.g. the activation of a sensor, but not its non-activation, which would correspond to a positive fragment of LTL. We could also want to syntactically force such an event to be ``good'' in the sense that if a specification is satisfied when a signal is off at some time, it should still be satisfied when the signal is on instead. It is therefore natural to ask whether a syntactic constraint on the positivity of $\LTL$ formulas could capture the semantic monotonicity, in the full setting or in some fragments corresponding to particular kinds of specifications.

We will also pay a close look at the two-variable fragment $\FOtw$ of $\FO$ and its $\LTL$ counterpart.
It was shown in \cite{PFO} that there exists a monotone $\FO$-definable language that is not definable in positive $\FO$. We give stronger variants of this counter-example language, and show that such a counter-example cannot be defined in $\FOtw[<]$.
This is obtained via a stronger result characterizing $\FOtw$-monotone in terms of positive fragments of bounded quantifier alternation.
We also give precise complexity results for deciding whether a regular language is monotone, refining results from \cite{PFO}.

The goal of this work is to understand at what point the phenomenon discovered in \cite{PFO} come into play: what are the necessary ingredients for such a counter-example ($\FO$-monotone but not $\FO$ positive) to exist? And on the contrary, which fragments of $\FO$ are better behaved, and can capture the monotonicity property with a semantic constraint, and allow for a decidable membership problem in the positive fragment. 


\subsection*{Outline and Contributions}

We begin by introducing two logical formalisms in \Cref{sec:FOandLTL}: First-Order Logic (\ref{sec:FO}) and Temporal Logic (\ref{sec:TL}).

Then, we lift some classical logical equivalences to positive logic in \Cref{sec:equiv}.
First we show that $\FOp$, $\FOthp$ and $\LTLp$ are equivalent in \Cref{equiv}.
We prove that the fragment $\FOtwp$ with (resp. without) successor predicate is equivalent to $\UTLp$ with (resp. without) $\X$ and $\Y$ operators available in \Cref{equiv2} (resp. \Cref{coro:equiv2withoutX}).

In \Cref{sec:monoids}, we give a characterisation of monotonicity using monoids (\Cref{ordrmono}) and we deduce from this an algorithm which decides the monotonicity of a regular language given by a monoid (\Cref{sec:algo}), completing the automata-based algorithms given in \cite{PFO}.
This leads us to the \Cref{complexity} which states that deciding the monotonicity of a regular language is in $\L$ when the input is a monoid while it is $\NL$-$\comp$ when the input is a DFA. This completes the previous result from \cite{PFO} showing $\PSPACE$-completeness for NFA input.

Finally, we study the relationship between semantic and syntactic positivity in \Cref{sec:semVSsynt}.
We give some refinements of the counter-example from \cite{PFO} (a regular and monotone language $\FO$-definable but not definable in $\FOp$).
Indeed, we show that the counter-example can be adapted to $\FOtw$ with the binary predicate "between" in \Cref{FO2be} and we show that we need only one predicate to find a counter-example in $\FO$ in \Cref{prop:onepred}.

We also consider a characterization of $\FOtw[<]$ from Thérien and Wilke \cite{OneQuantifierAlternation} stating that $\FOtw[<]$ is equivalent to $\Sigma_2 \cap \Pi_2$ where $\Sigma_2$ and $\Pi_2$ are fragments of $\FO$ with bounded quantifier alternation.
We show that $\FOtw$-monotone is characterized by $\Sigma_2^+ \cap \Pi_2^+$.

At last, we show that no counter-example for $\FO$ can be found in $\FOtw$ (without successor available) in \Cref{FO2clot}.
We conclude by leaving open the problem of expressive equivalence between $\FOtwp$ and $\FOtw$-monotone, as well as decidability of membership in $\FOtwp$ for regular languages (see \Cref{conj}).

\section{\texorpdfstring{$\FO$ and $\LTL$}{FO and LTL}} \label{sec:FOandLTL}

We work with a set of atomic unary predicates $\Sigma = \{a_1,a_2,...a_{|\Sigma|}\}$, and consider the set of words on alphabet $\alph$.
To describe a language on this alphabet, we use logical formulas. Here we present the different logics and how they can be used to define languages.

\subsection{First-order logics} \label{sec:FO}

Let us consider a set of binary predicates, $=$, $\neq$, $\leq$, $<$, $\succ$ and $\nsucc$, which will be used to compare positions in words. We define the subsets of predicates $\bin_0 := \{ \leq,<, \succ, \nsucc\}$, $\bin_< := \{\leq,<\}$ and $\bin_{\succ} := \{=, \neq, \succ, \nsucc\}$, and a generic binary predicate is denoted $\b$.
As we are going to see, equality can be expressed with other binary predicates in $\bin_0$ and $\bin_<$ when we have at least two variables.
This is why we do not need to impose that $=$ belongs to $\bin_0$ or $\bin_<$.
The same thing stands for $\neq$. Generally, we will always assume that predicates $=$ and $\neq$ are expressible.

Let us start by defining first-order logic $\FO$:

\begin{definition}
    Let $\bin$ be a set of binary predicates. The grammar of $\FO[\bin]$ is as follows:
    $$\varphi, \psi::= \bot \mid \top \mid \b(x,y) \mid a(x) \mid \varphi \land \psi \mid \varphi \lor \psi \mid \exists x, \varphi \mid \forall x, \varphi \mid  \neg \varphi,$$

    where $\b$ belongs to $\bin$.
\end{definition}

Closed $\FO$ formulas (those with no free variable) can be used to define languages. Generally speaking, a pair consisting of a word $u$ and a function $\val$ from the free (non-quantified) variables of a formula $\varphi$ to the positions of $u$ satisfies $\varphi$ if $u$ satisfies the closed formula obtained from $\varphi$ by replacing each free variable with its image by $\val$.

\begin{definition}
    Let $\varphi$, a formula with $n$ free variables, $x_1$, ..., $x_n$, and $u$ a word. Let $\val$ be a function of $\{x_1,...,x_n\}$ in $[\![0,|u|-1]\!]$. We say that $(u,\val)$ satisfies $\varphi$, and we define $u,\val \models \varphi$ by induction on $\varphi$ as follows:
    \begin{itemize}
        \item $u,\val \models \top$ and we never have $u, \val \models \bot$,
        \item $u, \val \models x < y$ if $\val(x) < \val(y)$,
        \item $u, \val \models x \leq y$ if $\val(x) \leq \val(y)$,
        \item $u, \val \models \succ(x,y)$ if $\val(y) = \val(x)+1$,
        \item $u, \val \models \nsucc(x,y)$ if $\val(y) \neq \val(x)+1$,
        \item $u, \val \models a(x)$ if $a \in u[\val(x)]$ \textbf{(note that we only ask inclusion here)},
        \item $u, \val \models \varphi \land \psi$ if $u, \val \models \varphi$ and $u, \val \models \psi$,
        \item $u, \val \models \varphi \lor \psi$ if $u, \val \models \varphi$ or $u, \val \models \psi$,
        \item $u, \val \models \exists x, \varphi(x,x_1,...,x_n)$ if there is $i$ of $u$ such that we have $u, \val \cup [x \mapsto i] \models \varphi$,
        \item $u, \val \models \forall x, \varphi(x,x_1,...,x_n)$ if for any index $i$ of $u$, $u, \val \cup [x \mapsto i] \models \varphi$,
        \item $u,\val \models \neg \varphi$ if we do not have $u,\val \models \varphi$.
    \end{itemize}

    For a closed formula, we simply note $u \models \varphi$.
\end{definition}

Here is an example:

\begin{example}
    The formula $\varphi = \exists x, \forall y, (x=y \lor \neg a(y))$ describes the set of non-empty words that admit at most one $a$. For example, $\{a\}\{a,b\}$ does not satisfy $\varphi$ because two of its letters contain an $a$, but $\{a,b,c\}\{b\}\emptyset$ does satisfy $\varphi$.
\end{example}

\begin{remark}
    The predicates $\succ$ and $\nsucc$ can be expressed in $\FOp[\bin_<]$ with three variables. If there are no restriction on variables, in particular if we can use three variables, all binary predicates in $\bin_0$ can be expressed from those in $\bin_<$. Thus, we will consider the whole set of binary predicates available when the number of variables is not constrained, and we will note $\FO$ for $\FO[\bin_0]$ or $\FO[\bin_<]$, which are equivalent, and similarly for $\FOp$.
\end{remark}

Let us now turn our attention to $\FOp$, the set of first-order formulas without negation. We recall definitions from \cite{PFO}.

\begin{definition}
    The grammar of $\FOp$ is that of $\FO$ without the last constructor, $\neg$.
\end{definition}


Let us also define monotonicity properties, starting with an order on words.

\begin{definition}
    A word $u$ is lesser than a word $v$ if $u$ and $v$ are of the same length, and for any index $i$ (common to $u$ and $v$), the $i$-th letter of $u$ is included in the $i$-th letter of $v$.
    When a word $u$ is lesser than a word $v$, we note $u \leqa v$.
\end{definition}

\begin{definition}
    Let $L$ be a language.
    We say that $L$ is monotone when for any word $u$ of $L$, any word greater than $u$ belongs to $L$.
\end{definition}

\begin{proposition}[\cite{PFO}]
        $\FOp$ formulas are monotone in unary predicates, i.e. if a model $(u,\val)$ satisfies a formula $\varphi$ of $\FOp$, and $u \leqa v$, then $(v,\val)$ satisfies $\varphi$.
\end{proposition}

We will also be interested in other logical formalisms, obtained either by restricting $\FO$, or several variants of temporal logics.

First of all, let us review classical results  obtained when considering restrictions on the number of variables. While an $\FO$ formula on words is always logically equivalent to a three-variable formula \cite{Kamp}, two-variable formulas describe a class of languages strictly included in that described by first-order logic. In addition, the logic $\FO$ is equivalent to Linear Temporal Logic (see below).

Please note: these equivalences are only true in the framework on word models. In other circumstances, for example when formulas describe graphs, there are formulas with more than three variables that do not admit equivalents with three variables or less.

\begin{definition}
    The set $\FOth$ is the subset of $\FO$ formulas using only three different variables, which can be reused. We also define $\FOthp$ for formulas with three variable and without negation. Similarly, we define $\FOtw$ and $\FOtwp$ with two variables.
\end{definition}

\begin{example}
    The formula $\exists y, \succ(x,y) \land (\exists x, b(x) \land (\forall z, z \geq x \lor z < y \lor a(z) ))$ (a formula with one free variable $x$ that indicates that the letter labeled by $x$ will be followed by a factor of the form $aaaaa. ..aaab$) is an $\FOth$ formula, and even an $\FOthp$ formula: there is no negation, and it uses only three variables, $x$, $y$ and $z$, with a reuse of $x$. On the other hand, it does not belong to $\FOtw$.
\end{example}

\subsection{Temporal logics} \label{sec:TL}

Some logics involve an implicit temporal dimension, where positions are identified with time instants.
For example, Linear Temporal Logic (LTL) uses operators describing the future, i.e. the indices after the current position in a word. This type of logic can sometimes be more intuitive to manipulate, and present better complexity properties, see introduction.  As mentioned above, $\FOtw$ is not equivalent to $\FO$. On the other hand, it is equivalent to $\UTL$, a restriction of $\LTL$ to its unary temporal operators. 

To begin with, let us introduce $\LTL$, which is equivalent to $\FO$.

\begin{definition}
    The grammar of $\LTL$ is as follows:
    \vspace{-1em}
    
    $$\varphi, \psi::= \bot \mid \top \mid a \mid \varphi \land \psi \mid \varphi \lor \psi \mid \X \varphi \mid \varphi \U \psi \mid \varphi \R \psi \mid \neg \varphi.
    $$
    Removing the last constructor gives the grammar of $\LTLp$.
\end{definition}

This logic does not use variables. To check that a word satisfies an $\LTL$ formula, we evaluate the formula at the initial instant, that is to say, the word's first position. The $\X$ constructor then describes constraints about the next instant, i.e. the following position in the word. So the word $a.u$, where $a$ is a letter, satisfies $\X \varphi$ if and only if the suffix $u$ satisfies $\varphi$. The construction $\varphi \U \psi$ ($\varphi$ until $\psi$) indicates that the formula $\psi$ must be verified at a given point in time and that $\varphi$ must be verified until then. We define $\varphi \R \psi$ as being equal to $\neg (\neg \varphi \U \neg \psi)$.
Let us define this formally:

\begin{definition}
    Let $\varphi$ be an $\LTL$ formula, and $u = u_0...u_{m-1}$ be a word. We say that $u$ satisfies $\varphi$ and define $u \models \varphi$ by induction on $\varphi$ as follows:
    \begin{itemize}
        \item $u \models \top$ and we never have $u \models \bot$,
        \item $u \models a$ if $a \in u[0]$,
        \item $u \models \varphi \land \psi$ if $u \models \varphi$ and $u \models \psi$,
        \item $u \models \varphi \lor \psi$ if $u \models \varphi$ or $u \models \psi$,
        \item $u \models \X \varphi$ if $u_1...u_{m-1} \models \varphi$,
        \item $u \models \varphi \U \psi$ if there is $i\in[\![0,m-1]\!]$ such that $u_i...u_{m-1} \models \psi$ and for all $j\in[\![0,i-1]\!]$, $u_j...u_{m-1} \models \varphi$,
        \item $u \models \varphi \R \psi$ if $u \models (\psi \U (\psi \wedge \varphi))$ or for all $i\in[\![0,m-1]\!]$ we have $u_i...u_{m-1} \models \psi$,
        \item $u \models \neg \varphi$ if we do not have $u \models \varphi$.
    \end{itemize}
\end{definition}

\begin{remark}
        Let us call $\varphi \XU \psi$ the formula $\X (\varphi \U \psi)$, for any pair $(\varphi,\psi)$ of $\LTL$ formulas. The advantage of $\XU$ is that $\X$ and $\U$ can be redefined from $\XU$. The notation $\U$ for $\XU$ is regularly found in the literature. 
\end{remark}

$\LTL$ is included in Temporal Logic, $\TL$. While the former speaks of the future, i.e. of the following indices in the word, thanks to $\X$, $\U$ and $\R$, the latter also speaks of the past. Indeed, we introduce $\Y$, $\S$ (since) and $\Q$ the respective past analogues of $\X$, $\U$ and $\R$.

\begin{definition}
    The grammar of $\TL$ is as follows:
    \vspace{-1em}
    
    $$\varphi, \psi::= \LTL \mid \Y \phi \mid \phi \S \psi \mid \varphi \Q \psi.$$
    Similarly, the grammar of $\TLp$ is that of $\LTLp$ extended with $\Y$, $\S$ and $\Q$.
\end{definition}

\begin{remark}
    As for $\XU$, we will write $\varphi \YS \psi$ for $\Y (\varphi \S \psi)$. We also note $\P \varphi$, $\F \varphi$, $\H \varphi$ and $\G \varphi$ for $\top \YS \varphi$, $\top \XU \varphi$, $\varphi \YS \bot$ and $\varphi \XU \bot$ respectively. The formulas $\F \varphi$ and $\G \varphi$ mean respectively that the formula $\varphi$ will be satisfied at least once in the future ($\F$ as Future), and that $\varphi$ will always be satisfied in the future ($\G$ as Global). Similarly, the operators $\P$ (as Past) and $\H$ are the respective past analogues of $\F$ and $\G$.
\end{remark}



When evaluating an $\LTL$ or $\TL$ formula on a word $u=u_0\dots u_m$, we start by default on the first position $u_0$. However, we need to define more generally the evaluation of a $\TL$ formula on a word from any given position:

\begin{definition}
    Let $\varphi$ be a $\TL$ formula, $u = u_0...u_{m-1}$  a word, and $i\in[\![0,m-1]\!]$. We define $u,i \models \varphi$ by induction on $\varphi$:
    \begin{itemize}
        \item $u,i \models \top$ and we never have $u \models \bot$,
        \item $u,i \models a$ if $a \in u_i$,
        \item $u,i \models \varphi \land \psi$ if $u,i \models \varphi$ and $u,i \models \psi$,
        \item $u,i \models \varphi \lor \psi$ if $u,i \models \varphi$ or $u,i \models \psi$,
        \item $u,i \models \X \varphi$ if $u, i+1 \models \varphi$,
        \item $u,i \models \varphi \U \psi$ if there is $j\in[\![i,m-1]\!]$ such that $u,j \models \psi$ and for all $k\in[\![i,j-1]\!]$, $u,k \models \varphi$,
        \item $u,i \models \psi \R \varphi$ if $u,i \models \neg (\neg \psi \U \neg \varphi)$,
        \item $u,i \models \neg \varphi$ if we do not have $u,i \models \varphi$,
        \item $u,i \models \Y \varphi$ if $u,i-1 \models \varphi$,
        \item $u,i \models \varphi \S \psi$ if there is $j\in[\![0,i]\!]$ such that $u,j \models \psi$ and for all $k\in[\![j+1,i]\!]$, $u,k \models \varphi$.
    \end{itemize}
\end{definition}

Finally, let us introduce $\UTL$ and $\UTLp$, the Unary Temporal Logic and its positive version. The $\UTL$ logic does not use the $\U$ or $\R$ operator, but only $\X$, $\F$ and $\G$ to talk about the future. Similarly, we cannot use $\S$ or $\Q$ to talk about the past.

\begin{definition}
    The grammar of $\UTL$ is as follows:
     $$
    \varphi, \psi::= \bot \mid \top \mid a \mid \varphi \land \psi \mid \varphi \lor \psi \mid \X \varphi \mid \Y \phi \mid \P \varphi \mid \F \varphi \mid \H \varphi \mid \G \varphi \mid \neg \varphi .
    $$

    
    We define define $\UTL[\P,\F,\H,\G]$ from this grammar by deleting the constructors $\X$ and $\Y$.

    
    The grammar of $\UTLp$ is obtained by deleting the last constructor, and similarly, we define $\UTLp[\P,\F,\H,\G]$ by deleting the negation in $\UTL[\P,\F,\H,\G]$.
\end{definition}

\begin{remark}
    In the above definition, $\H$ and $\G$ can be redefined with $\P$ and $\F$ thanks to negation, but are necessary in the case of $\UTLp$.
\end{remark}

When two formulas $\varphi$ and $\psi$ are logically equivalent, i.e. admit exactly the same models, we denote it by $\varphi \equiv \psi$. Note that a closed  $\FO$ formula can be equivalent to an $\LTL$ formula, since their models are simply words. Similarly, we can have $\varphi\equiv\psi$ when $\varphi$ is an $\FO$ formula with one free variable (having models of the form $(u,i)$) and $\psi$ is a $\LTL$ or $\TL$ formula, this time not using the default starting position for $\TL$ semantics.

\section{Logical equivalences} \label{sec:equiv}

We want to lift to positive fragments some classical theorems of equivalence between logics, such as these classical results:

\begin{theorem}[\cite{Kamp,UTL}]$ $
\begin{itemize}
\item $\FO$ and $\LTL$ define the same class of languages.
\item $\FOtw$ and $\UTL$ define the same class of languages.
\end{itemize}
\end{theorem}

\subsection{\texorpdfstring{Equivalences to $\FOp$}{Equivalences to FO+}} \label{sec:equivFO+}

We aim at proving the following theorem, lifting classical results from $\FO$ to $\FOp$:

\begin{theorem}\label{equiv}
    The logics $\FOp$, $\LTLp$ and $\FOthp$ describe the same languages.
\end{theorem}

\begin{lemma}\label{LTLtoFO}
    The set of languages described by $\LTLp$ is included in the set of languages recognised by $\FOthp$.
\end{lemma}

The proof is direct, see Appendix \ref{app:LTLtoFO} for details. From $\LTLp$ to $\FOp$, we can interpret in $\FOp$ all constructors of $\LTLp$.



Let us introduce definitions that will be used in the proof of the next lemma.

\begin{definition}
    Let $\qr(\varphi)$ be the quantification rank of a formula $\varphi$ of $\FOp$ defined inductively by:
    \begin{itemize}
        \item if $\varphi$ contains no quantifier then $\qr(\varphi) = 0$,
        \item if $\varphi$ is of the form $\exists x, \psi$ or $\forall x, \psi$ then $\qr(\varphi) = \qr(\psi) + 1$,
        \item if $\varphi$ is of the form $\psi \lor \chi$ or $\psi \land \chi$ then $\qr(\varphi) = \max(\qr(\psi),\qr(\chi))$.
    \end{itemize}
\end{definition}

\begin{definition}
    A separated formula is a positive Boolean combination of purely past formulas (which do not depend on the present and future), purely present formulas (which do not depend on the past and future) and purely future formulas (which do not depend on the past and present).
\end{definition}

We will adapt previous work to show the following auxiliary result:

\begin{lemma}\label{lem:separation}
    Let $\varphi$ be a $\TLp$ formula with possible nesting of past and future operators. There is a separated formula of $\TLp$ that is equivalent to $\varphi$.
\end{lemma}

\begin{proof}[Proof (Sketch)]

    Our starting point is the proof given by Kuperberg and Vanden Boom in \cite[lemma 5]{FOtoLTL}, which proves the equivalence between generalisations of the logics $\FO$ and $\LTL$, to the so-called cost $\FO$ and cost $\LTL$. When specialised to $\FO$ and $\LTL$, this corresponds to the case where negations appear only at the leaves of formulas. This brings us closer to our goal.

    First of all, \cite{FOtoLTL} proves a generalised version of the separation theorem from \cite{sep}.
    In \cite{sep}, it is proven that any formula of $\TL$ is equivalent to a separated formula, and a particular attention to positivity is additionally given in \cite{FOtoLTL}.
    Indeed \cite{FOtoLTL} also shows that such a Boolean combination can be constructed while preserving the formula's positivity.
    One can also check \cite{OliveiraR21} to verify that positivity of a formula is kept when separating the formula.
    Thus, a formula in $\TLp$ can be written as a Boolean combination of purely past, present and future formulas themselves in $\TLp$.
\end{proof}

Now we are ready to show the main result of this section:

\begin{lemma}\label{FOtoLTL}
    The set of languages described by $\FOp$ is included in the set of languages recognised by $\LTLp$.
\end{lemma}

\begin{proof}

    We follow \cite{FOtoLTL}, which shows a translation from $\FO$ to $\TL$ by induction on the quantification rank. We have adapted this to suit our needs.

    Let $\varphi(x)$ be an $\FOp$ formula with a single free variable. Let us show by induction on $\qr(\varphi)$ that $\varphi$ is equivalent to a formula of $\TLp$.

    \underline{Initialisation:}

    If $\qr(\varphi)$ is zero, then $\varphi(x)$ translates directly into the $\TLp$ formula. Indeed, disjunctions and conjunctions translate immediately into $\TLp$. Furthermore, unary predicates of the form $a(x)$ translate into $a$ and binary predicates trivialize into $\top$ and $\bot$ (e.g. $x<x$ translates into $\bot$ and $x=x$ into $\top$). For example, $(x \leq x \land a(x)) \lor (b(x) \land c(x)) \lor x < x$ translates into $(\top \land a) \lor (b \land c) \lor \bot$.

    \underline{Heredity:}

    Suppose that any $\FOp$ free single-variable formula of quantification rank strictly less than $\qr(\varphi)$ translates into a $\TLp$ formula, and $\qr(\varphi)$ is strictly positive.

    If $\varphi$ is a disjunction or conjunction, we need to transform its various clauses.
    So, without loss of generality, let us assume that $\varphi(x)$ is of the form $\exists y, \psi(x,y)$ or $\forall y, \psi(x,y)$.

    Let us denote $a_1$, ... $a_n$ where $n$ is a natural number, the letters (which are considered as unary predicates) in $\psi(x,y)$ applied to $x$.

    For any subset $S$ of $[\![1,n]\!]$, we note $\psi^S(x,y)$ the formula $\psi(x,y)$ in which each occurrence of $a_i(x)$ is replaced by $\top$ if $i$ belongs to $S$ and by $\bot$ otherwise, for any integer $i$ of $[\![1,n]\!]$.

    We then have the logical equivalence:
    $$
    \psi(x,y) \equiv \bigvee_{S \subseteq [\![1,n]\!]}
    \left(
        \bigwedge_{i \in S} 
        a_i(x) \land \bigwedge_{i \notin S} \neg a_i(x)
        \land \psi^S(x,y)
    \right).
    $$

    We are going to show that the negations in the above formula are optional.
    Let us note:
    $$
    \psi^+(x,y) \equiv \bigvee_{S \subseteq [\![1,n]\!]}
    \left(
        \bigwedge_{i \in S} a_i(x)
    \land \psi^S(x,y)
    \right).
    $$

    Let us then show the equivalence of the formulas $\psi(x,y)$ and $\psi^+(x,y)$ using the monotonicity of $\psi$ as an $\FOp$ formula. 
    First of all, it is clear that any model satisfying $\psi(x,y)$ satisfies $\psi^+(x,y)$.

    Conversely, suppose $\psi^+(x,y)$ is satisfied. We then have a subset $S$ of $[\![1,n]\!]$ such that $(\wedge_{i \in S} a_i(x))
    \land \psi^S(x,y)$ is satisfied. In particular, according to the values taken by the unary predicates in $x$, there exists a subset $S'$ of $[\![1,n]\!]$ containing $S$ such that $(\wedge_{i \in S'} a_i(x)) \land (\wedge_{i \notin S'} \neg a_i(x))
    \land \psi^S(x,y)$ is satisfied. Now, $\psi$ is monotone in the different predicates $a_1$,...,$a_n$. So $(\wedge_{i \in S'} a_i(x)) \land (\wedge_{i \notin S'} \neg a_i(x))
    \land \psi^{S'}(x,y)$ is also satisfied, and $\psi(x,y)$ is therefore satisfied.

    The rest of the proof is similar to the proof from \cite{FOtoLTL}:
    the quantifiers on $y$ commute with the disjunction on $S$ and the conjunction on $i$ of the formula $\psi^+$. We can therefore fix a subset $S$ of $[\![1,n]\!]$ and simply consider $\exists y, \psi^S(x,y)$ or $\forall y, \psi^S(x,y)$.
    We then replace $\psi^S(x,y)$ with a formula that depends only on $y$ by replacing each binary predicate of the form $\b(x,z)$ with a unary predicate $\P_{\b}(z)$.
    For example, we can replace $x<z$, $z<x$ or $x=z$ by a unary predicate $\P_>(z)$, $\P_<(z)$ or $\P_=(z)$. 
    We then obtain a formula $\psi'^{S}(y)$ on which we can apply the induction hypothesis (since there is only one free variable). This yields a formula $\chi$ from $\TLp$, equivalent to $\psi'^{S}(y)$ and we have:
    $$
    \exists y, \psi^{S}(x,y) \equiv \P \chi \lor \chi \lor \F \chi,
   \text{~~~and~~~}
    \forall y, \psi^{S}(x,y) \equiv \H \chi \land \chi \land \G \chi.
    $$

    Let $\chi'$ be the formula obtained ($\P \chi \lor \chi \lor \F \chi$ or $\H \chi \land \chi \land \G \chi$).
    The resulting formula $\chi'$ then involves unary predicates of the form $\P_{\b}$. We then use \Cref{lem:separation} to transform $\chi'$ into a positive Boolean combination of purely past, present and future positive formulas, where predicates $\P_\b$ trivialize into $\top$ or $\bot$.
    For example, $\P_<$ trivializes into $\top$ in purely past formulas, into $\bot$ in purely present or future formulas.

    This completes the induction. From a formula in $\FOp$, we can construct an equivalent formula in $\TLp$.

    Ultimately, we can return to a future formula. Indeed, we want to evaluate in $x=0$, so the purely past formulas, isolated by the separation lemma (\Cref{lem:separation}), trivialize into $\bot$ or $\top$.

    Now, to translate a closed formula $\varphi$ from $\FOp$ to $\LTLp$, we can add a free variable by setting $\varphi'(x) = \varphi \land (x=0)$. Then, by the above, $\varphi'$ translates into a formula $\chi$ from $\LTLp$, logically equivalent to $\varphi$.

\end{proof}
We can now turn to the proof of \Cref{equiv}.

\begin{proof}[Proof of \Cref{equiv}]
    By \Cref{LTLtoFO}, we have the inclusion of the languages described by $\LTLp$ in those described by $\FOthp$, which is trivially included in $\FOp$. By \Cref{FOtoLTL}, the converse inclusion of $\FOp$ into $\LTLp$ holds. So we can conclude that the three logical formalisms are equi-expressive.
\end{proof}

\subsection{\texorpdfstring{Equivalences in fragments of $\FOp$}{Equivalences in fragments of FO+}} \label{sec:equivfrag}

\begin{theorem}\label{equiv2}
    The languages described by $\FOtwp[\bin_0]$ formulas with one free variable are exactly those described by $\UTLp$ formulas.
\end{theorem}

\begin{proof}
    First, let us show the $\UTLp$ to $\FOtwp$ direction. In the proof of \Cref{LTLtoFO}, as is classical, three variables are introduced only when translating $\U$. By the same reasoning as for $\X$, it is clear that translating $\Y$ introduces two variables. It remains to complete the induction of \Cref{LTLtoFO} with the cases of $\P$, $\F$, $\H$ and $\G$, but again we can restrict ourselves to future operators by symmetry:
    
    \begin{itemize}
        \item $[\F\varphi](x) = \exists y, x < y \land [\varphi](y)$ ;
        \item $[\G\varphi](x) = \forall y, y \leq x \lor [\varphi](y)$.
    \end{itemize}

    For the converse direction from $\FOtwp$ to $\UTLp$, we draw inspiration from \cite[Theorem 1]{UTL}. This proof is similar to that of \cite{FOtoLTL} used previously in the proof of \Cref{FOtoLTL}: we perform a disjunction on the different valuations of unary predicates in one free variable to build a formula with one free variable. However, the proof of \Cref{FOtoLTL} cannot be adapted as it is, since it uses the separation theorem which does not preserve the membership of a formula to $\UTL$, see \cite[Lem 9.2.2]{tempLogicBook}. However, the article \cite{UTL} uses negations and we must therefore construct our own induction case for the universal quantifier that is treated in \cite{UTL} via negations.

    The beginning of the proof is identical to that of \Cref{FOtoLTL}. Using the same notations, let us consider a formula $\psi^S(x,y)$ with no unary predicate applied to $x$. We cannot directly replace binary predicates with unary predicates, because this relied on the separation theorem.

    Let us consider, as in \cite{UTL}, the \emph{position formulas}, $y<x \land \nsucc(y,x)$, $\succ(y,x)$, $y=x$, $\succ(x,y)$ and $x < y \land \nsucc(x,y)$, whose set is denoted $\Tau$.

    We then have the logical equivalence:

    $$
    \psi^S(x,y)
    \equiv \bigvee_{\tau \in \Tau} \tau(x,y) \land \psi_{\tau}^S(y)
    \equiv \bigwedge_{\tau \in \Tau} \tau(x,y) \implies \psi_{\tau}^S(y),
    $$

    where $\psi_{\tau}^S(y)$ is obtained from the formula $\psi^S(x,y)$ assuming the relative positions of $x$ and $y$ are described by $\tau$. The above equivalence holds because $\Tau$ forms a partition of the possibilities for the relative positions of $x$ and $y$: exactly one of the five formulas $\tau(x,y)$ from $\Tau$ must hold. Since $x$ and $y$ are the only two variables, any binary predicate involving $x$ is a binary predicate involving $x$ and $y$ (or else it involves only $x$ and is trivial). Binary predicates are therefore trivialized according to the position described by $\tau$.
    
    \begin{example}
        For $\psi^S(x,y) = \nsucc(x,y) \land a(y) \land (\forall x, x \leq y \lor b(y))$ and for the position formula $\tau = y<x \land \nsucc(y, x)$, we have $\psi_{\tau}^S(y) = \top \land a(y) \land (\forall x, x \leq y \lor b(y))$. We do not replace the bound variable $x$. We have obtained a formula with one free variable, so we can indeed use the induction hypothesis.
    \end{example}

    We use disjunction in the case of an existential quantifier (as in \cite{UTL}) and conjunction in the case of a universal quantifier.
    We then need to translate $\exists y, \tau(x,y) \land \psi_{\tau}^S(y)$ and $\forall y, \tau(x,y) \implies \psi_{\tau}^S(y)$, which we note respectively $[\tau]_{\exists}$ and $[\tau]_{\forall}$, in $\UTLp$, for any position formula $\tau$. For readability we omit $\psi_{\tau}^S$ in this notation, but $[\tau]_{\exists}$ and $[\tau]_{\forall}$ will depend on $\psi_{\tau}^S$.
    In each case, we note $\chi$ for the $\UTLp$ formula obtained by induction from $\psi_{\tau}^S(y)$:
    $$
    \def\arraystretch{1.5}
    \begin{array}{c}
    [y<x \land \nsucc(y,x)]_{\exists} \equiv \Y \P \chi,\\   
    ~[y<x \land \nsucc(y,x)]_{\forall} \equiv \Y \H \chi,~\\ 
    ~ [\succ(y,x)]_{\exists} \equiv \succ(y,x)]_{\forall} \equiv \Y \chi,~\\ 
     ~[y=x]_{\exists} \equiv [y=x]_{\forall} \equiv \chi,    \\ 
    ~[\succ(x,y)]_{\exists} \equiv \succ(x,y)]_{\forall} \equiv \X \chi,~\\ 
   ~[x < y \land \nsucc(x,y)]_{\exists} \equiv \X \F \chi,~\\ 
    ~[x < y \land \nsucc(x,y)]_{\forall} \equiv \X \G \chi.~
    \end{array}
  $$ 
\end{proof}

\begin{corollary} \label{coro:equiv2withoutX}
    The logic $\FOtwp[\bin_<]$ is equivalent to $\UTLp[\P,\F,\H,\G]$.
\end{corollary}

\begin{proof}
    For the right-to-left direction, it suffices to notice that the predicates used to translate the constructors of $\UTLp[\P,\F,\H,\G]$ in the previous proof belong to $\bin_<$.

    For the left-to-right direction, simply replace the set $\Tau$ in \Cref{equiv2} proof by $\Tau' = \{ y<x, y=x, x<y \}$. Once again, we obtain an exhaustive system of mutually exclusive position formulas that allow us to trivialize binary predicates. The proof of \Cref{equiv2} can thus be lifted immediately to this case.
\end{proof}

We showed that several classical logical equivalence results can be transposed to their positive variants.

\section{Characterisation of monotonicity} \label{sec:monoids}

So far, we have focused on languages described by positive formulas, from which monotonicity follows. Here, we focus on the monotonicity property and propose a characterisation. We then derive a monoid-based algorithm that decides, given a regular language $L$, whether it is monotone, refining results from \cite{PFO} focusing on automata-based algorithms.

\subsection{Characterisation by monoids} \label{sec:charac}

We assume the reader familiar with monoids (see Appendix \ref{ap:defmon} for detailed definitions).

We will note $(\M,\cdot)$ a monoid and $\M_L$ the syntactic monoid of a regular language $L$ and $\leq_L$ the syntactic order.


\begin{lemma}\label{ordrmono}
    Let $L\subseteq \alph^*$ be a regular language. Then $L$ is monotone if and only if there is an order $\leqml$ on $\M_L$ compatible with the product $\cdot$ and included in $\leq_L$ which verifies:
    $$
    \forall (u,v) \in \alph^* \times \alph^*, u \leqa v \implies h(u) \leqml h(v),
    $$

    where $h$ denotes the canonical projection.

\end{lemma}

The proof is left in \Cref{ap:proofordrmono}.



\begin{theorem}\label{caracmonot}
    Let $L\subseteq\alpha^*$ be a regular language, and $\leq_L$ be its syntactic order. The language $L$ is monotone if and only if we have:
    $$
    \forall (s,s') \in \alph^2, s \subseteq s' \implies h(s) \leq_L h(s'),
    $$
    where $h:\alph^*\to M_L$ denotes the canonical projection onto the syntactic monoid.
\end{theorem}

\begin{proof}
    For the left-to-right direction let $L$ be a monotone language and $s \subset s'$.
    Let $m$ and $n$ be two elements of $\M_L$ such that $mh(s)n \in h(L)$.
    Since $h : \alph^* \to \M_L$ is surjective, let $u \in h^{-1}(m)$ and $v \in h^{-1}(n)$.
    Then $usv \in L$ since $h$ recognises $L$.
    So $us'v \in L$ by monotonicity of $L$.
    Thus $mh(s')n \in h(L)$.
    We can conlude that $h(s) \leq_L h(s')$.


    For the converse direction, suppose that $\leq_L$ verifies the condition of \Cref{caracmonot}.
    We can remark that $\leq_L$ is compatible with the product of the monoid.
    Therefore, the conditions of \Cref{ordrmono} are verified by $\leq_L$.
    
\end{proof}

\subsection{An algorithm to decide monotonicity} \label{sec:algo}

We immediately deduce from \Cref{caracmonot} an algorithm for deciding the monotonicity of a regular language $L$ from its syntactic monoid. Indeed, it is sufficient to check for any pair of letters $(s,s')$ such that $s$ is included in $s'$ whether $m \op h(s) \op n \in h(L)$ implies $m \op h(s') \op n \in h(L)$ for any pair $(m,n)$ of elements of the syntactic monoid, where $h$ denotes the canonical projection onto the syntactic monoid.

This algorithm works for any monoid that recognises $L$ through a surjective $h:\alph^*\to M$, not just its syntactic monoid.
Indeed, for any monoid, we start by restricting it to $h(\alph^*)$ to guarantee that $h$ is surjective.
Then, checking the above implication is equivalent to checking whether $s \leq_L s'$ for all letters $s$ and $s'$ such that $s$ is included in $s'$.




        

        




    


This is summarised in the following proposition:

\begin{proposition} \label{prop:algo}
    There is an algorithm which takes as input a monoid $(\M,\op)$ recognising a regular language $L$ through a morphism $h$ and decides whether $L$ is monotone in $O(|\alph|^2|\M|^2)$.
\end{proposition}

It was shown in \cite[Thm 2.5]{PFO} that deciding monotonicity is $\PSPACE$-complete if the language is given by an NFA, and in \textsc{P} if it is given by a DFA.

We give a more precise result for DFA, and give also the complexity for monoid input:

\begin{proposition}\label{complexity}
    Deciding whether a regular language is monotone is in $\L$ when the input is a monoid while it is $\NL-\comp$ when it is given by a DFA.
\end{proposition}

See Appendix \ref{ap:proofmon} for the proof.

\section{Semantic and syntactic monotonicity} \label{sec:semVSsynt}

The paper \cite[Definition 4.2]{PFO} exhibits a monotone language definable in $\FO$ but not in $\FOp$. The question then arises as to how simple such a counter-example can be. For instance, can it be taken in specific fragments of $\FO$, such as $\FOtw$. This section presents a few lemmas that might shed some light on the subject, followed by some conjectures.

From now on we will write $A$ the alphabet $\alph$.


\subsection{Refinement of the counter-example in the general case} \label{sec:refinement}

In \cite{PFO}, the counter-example language that is monotone and $\FO$-definable but not $\FOp$-definable uses three predicates $a$, $b$ and $c$ and is as follows:
$$
K = ((abc)^*)^{\uparrow} \cup A^* \top A^*.
$$

It uses the following words to find a strategy for Duplicator in $\EF_k^+$:
$$
u_0 = (abc)^n
\text{ and }
u_1 = \left(\binom{a}{b}\binom{b}{c}\binom{c}{a}\right)^n \binom{a}{b} \binom{b}{c},
$$

where $n$ is greater than $2^k$, and $\binom{s}{t}$ is just a compact notation for the letter $\{s,t\}$ for any predicates $s$ and $t$.

This in turns allows to show the failure on Lyndon's preservation theorem on finite structures \cite{PFO}.
Our goal in this section is to refine this counter-example to more constrained settings. We hope that by trying to explore the limits of this behaviour, we achieve a better understanding of the discrepancy between monotone and positive.

In Section \ref{sec:between}, we give a smaller fragment of $\FO$ where the counter-example can still be encoded.
In Section \ref{sec:onepred}, we show that the counter-example can still be expressed with a single unary predicate. This means that it could occur for instance in $\LTLp$ where the specification only talks about one sensor being activated or not.

\subsubsection{Using the between predicate}\label{sec:between}

First, let us define the ``between'' binary predicate introduced in \cite{between}.

\begin{definition}\cite{between}
    For any unary predicate $a$ (not only predicates from $\Sigma$ but also Boolean combination of them), $a$ also designates a binary predicate, called between predicate, such that for any word $u$ and any valuation $\val$, $(u,\val) \models a(x,y)$ if and only if there exists an index $i$ between $\val(x)$ and $\val(y)$ excluded such that $(u_i,\val) \models a$, where $u_i$ is the $i$-th letter of $u$.

    We denote $\be$ the set of between predicates and $\be^+$ the set of between predicates associated to the set of positive unary predicates.
\end{definition}

Is is shown in \cite{between} that $\FO^2[\bin_0\cup\be]$ is strictly less expressive than $\FO$.

\begin{proposition} \label{FO2be}
    There exists a monotone language definable in $\FOtw[\bin_0 \cup \be]$ which is not definable in $\FOtwp[\bin_0 \cup \be^+]$.
\end{proposition}

\begin{proof}
    We can use the same words $u_0$ and $u_1$ defined above with the following language:
    $$
    K \cup A^*\left(
        \binom{a}{b}^2 \cup \binom{b}{c}^2 \cup \binom{c}{a}^2 \cup
        \binom{a}{b}\binom{c}{a} \cup
        \binom{b}{c}\binom{a}{b} \cup
        \binom{c}{a}\binom{b}{c}
    \right)A^* .
    $$

    Indeed, in \cite{PFO}, it is explained that we need to look for some ``anchor position'' to know whether a word belongs to $K$. Such positions resolve the possible ambiguity introduced by double letters of the form $\binom{a}{b}$, that could play two different roles for witnessing membership in $((abc)^*)^{\uparrow}$.
    Indeed, if $\binom{a}{b}$ appears in a word, we cannot tell whether it stands for an $a$ or a $b$.
    In contrast, anchor letters have only one possible interpretation.
    They may be singletons ($\{a\}$, $\{b\}$, $\{c\}$) or consecutive double letters such as $\binom{a}{b}\binom{c}{a}$ which can only be interpreted as $bc$.
    Here, we accept any word containing an anchor of the second kind.
    This means that in remaining words we will only be interested in singleton anchors.
    Thus, we need two variables only to locate consecutive anchors and between predicates to check if the letters between the anchors are double letters.
    See Appendix \ref{ap:formula} for a more detailed description of a formula.

\end{proof}


\subsubsection{Only one unary predicate}\label{sec:onepred}

Now, let us show another refinement.
We can lift $K$ to a counter-example where the set of predicates $\Sigma$ is reduced to a singleton.

\begin{proposition} \label{prop:onepred}
    As soon as there is at least one unary predicate, there exists a monotone language definable in $\FO$ but not in $\FOp$.
\end{proposition}

\begin{proof}

Suppose $\Sigma$ reduced to a singleton.
Then, $A$ is reduced to two letters which we note $0$ and $1$ with $1$ greater than $0$.
We will encode each predicate from $\{a,b,c\}$ and a new letter $\#$ (the separator) into $A^*$ as follows:
$$
\left\{
\begin{array}{l}
    [a] = 001 \\ \relax
    [b] = 010 \\ \relax
    [c] = 100 \\ \relax
    [\#] = 100001
\end{array}
\right. .
$$

Thus, the letter $\binom ab$ will be encoded by $[ab]=011$, etc.
We will encode the language $K$ as follows:
$$
[K] = (([a][\#][b][\#][c][\#])^*)^{\uparrow} \cup A^* 1 (A^4 \backslash 0^4) 1 A^* \cup A^*1^5A^*.
$$

First, we can notice that $[K]$ is monotone.

Let us show how the separator $[\#]$ is used.
Let $w$ be a word over $A^*$.
If $w$ contains a factor of the form $1u1$ where $u$ is a word of $4$ letters containing the letter $1$, then $w$ immediately belongs to $[K]$.
This is easy to check with an $\FO$-formula so we can suppose that $w$ does not contain such a factor.
Similarly, we can suppose that $1^5$ (corresponding to $\top$ in the original $K$) is not a factor of $w$.
Then, it is easy to locate a separator since $100001$ will always be a separator factor.
Therefore, we can locate factor coding letters in $w$.
Then we can do the same thing as \cite{PFO} to find an $\FO$-formula: we have to fix some anchors (factors coding letters whose roles are not ambiguous as explained in the proof of \Cref{FO2be}) and check whether they are compatible.
For example, suppose $w$ contains a factor of the form $[a][\#]([ab][\#][bc][\#][ca][\#])^n[bc]$.
Then $[a]$ is an anchor.
The last factor $[ca][\#][bc]$ is also an anchor since it can only be interpreted as $([a][\#][b])^\uparrow$.
Since there are no anchors in between $[a]$ and $[bc]$ we just have to verify their compatibility.
Here it is the case:
in between the anchors, each $[ab]$ can be interpreted as $[b]^\uparrow$, $[bc]$ as $[c]^\uparrow$ and $[ca]$ as $[a]^\uparrow$.
If we were to replace $[a]$ with $[c]$, $[c]$ would still be an anchor but would not be compatible with $[bc]$.
This achieves the description of an $\FO$-formula for $[K]$.

Furthermore, it is not $\FOp$-definable.
Indeed, let $k\in\N$ be an arbitrary number of rounds for an $\EF^+$-game.
We can choose $n > 2^k$ such that Duplicator has a winning strategy for $u_0$ and $u_1$ defined as follows:
$$
[u_0] = ([a][\#][b][\#][c][\#])^n
\text{ and }
[u_1] = ([ab][\#][bc][\#][ca][\#])^n[ab][\#],
$$

where $[ab] = 011$, $[bc] = 110$ and $[ca] = 101$.

We can adapt the strategy for $u_0$ and $u_1$ (from \cite[Lemma 4.4]{PFO}) to $[u_0]$ and $[u_1]$.
For example, if Spoiler plays the $i$-th letter of a factor $[bc]$, then it is similar to playing the letter $\binom{b}{c}$ in $u_1$.
Thus, if Duplicator answers by playing the $j$-th $b$ or $c$ in $u_0$, then he should answer by playing the $i$-th letter of the $j$-th $[b]$ or $[c]$ respectively, for any natural integers $i$ and $j$.
In the same way, if Spoiler plays in a separator character, then Duplicator should answer by playing the same letter of the corresponding separator character in the other word according to the strategy.

\end{proof}

\subsection{Stability through monotone closure} \label{sec:stability}

It has been shown by Thérien and Wilke \cite{OneQuantifierAlternation} that languages $\FOtw[\bin_<]$-definable are exactly those who are both $\Sigma_2$-definable and $\Pi_2$-definable where $\Sigma_2$ is the set of $\FO$-formulas of the form $\exists x_1,....,x_n \forall y_1,...,y_m \varphi(x_1,...,x_n,y_1,...y_m)$ where $\varphi$ does not have any quantifier and $\Pi_2$-formulas are negations of $\Sigma_2$-formulas.
Hence, $\Sigma_2 \cup \Pi_2$ is the set of $\FO$-formulas in prenex normal form with at most one quantifier alternation.
Moreover, Pin and Weil \cite{PolynomialClosure} showed that $\Sigma_2$ describes the unions of languages of the form $A_0^*.s_0.A_1^*.s_1.....s_t.A_{t+1}^*$, where $t$ is a natural integer, $s_i$ are letters from $A$ and $A_i$ are subalphabets of $A$.

Even though we do not know yet whether $\FOtwp$ captures the set of monotone $\FOtw$-definable languages, we can state the following theorem:

\begin{theorem} \label{thm:sigpimon}
    The set $\Sigma_2^+ \cap \Pi_2^+$ of languages definable by both positive $\Sigma_2$-formulas (written $\Sigma_2^+$) and positive $\Pi_2$-formulas (written $\Pi_2^+$) is equal to the set of monotone $\FOtw$-definable languages.
\end{theorem}

In order to prove \Cref{thm:sigpimon}, we shall introduce a useful definition:

\begin{definition}
    For any language $L$, we write $L^{\dualmon} = ((L^c)^\downarrow)^c$ the dual closure of $L$, where $L^c$ stands for the complement of $L$ and $L^\downarrow$ is the downwards monotone closure of $L$.
\end{definition}

\begin{remark}\label{rmk:monclos}
    It is straightforward to show that $L^{\dualmon}$ is the greatest monotone language included in $L$ for any language $L$.
    In particular, a monotone language is both equal to its monotone closure and its dual monotone closure.
\end{remark}

Now, let us show the following lemma:

\begin{lemma}\label{lem:sig2+}
    The set $\Sigma_2^+$ captures the set of monotone $\Sigma_2$-definable languages.
\end{lemma}

\begin{proof}
    First, it is clear that $\Sigma_2^+$ describes monotone $\Sigma_2$-definable languages.

    Next, it is enough to show that the monotone closure of a $\Sigma_2$-definable language is $\Sigma_2^+$-definable.

    So let us consider a $\Sigma_2$-definable language $L$.
    Since a disjunction of $\Sigma_2^+$ formulas is equivalent to a $\Sigma_2^+$ formula, we can suppose thanks to \cite{PolynomialClosure} that $L$ is of the form $A_0^*.s_0.A_1^*.s_1.....s_t.A_{t+1}^*$ as explained above.

    Therefore, $L^\uparrow$ is described by the following $\Sigma_2^+$-formula:
    $$
        \exists x_0,...,x_t, \forall y, x_0 < ... < x_t \land \bigwedge_{i=0}^t s_i(x_i)
        \land \bigwedge_{i=0}^{t+1} ( x_{i-1}<y<x_i \Rightarrow A_i(y)),
    $$

    where $B(x)$ means $\bigvee_{b \in B} b(x)$ for any subalphabet $B$, $x_{-1} < y < x_0$ means $y<x_0$ and $x_t < y < x_{t+1}$ means $x_t < y$.

    
\end{proof}

This immediately gives the following lemma which uses the same sketch proof:

\begin{lemma}\label{lem:sig2-}
    The set $\Sigma_2^-$ ($\Sigma_2$-formulas with negations on all predicates) captures the set of downwards closed $\Sigma_2$-definable languages.
\end{lemma}

We can now deduce the following lemma:

\begin{lemma}\label{pi2+}
    The set $\Pi_2^+$ captures the set of monotone $\Pi_2$-definable languages.
\end{lemma}

\begin{proof}
    Then again, we only need to show the difficult direction.

    Let $L$ be a $\Pi_2$-definable language.
    It is enough to show that $L^\dualmon$ is $\Pi_2^+$-definable according to \Cref{rmk:monclos}.

    By definition of $\Pi_2$, the complement $L^c$ of $L$ is $\Sigma_2$-definable.
    Hence, $(L^c)^\downarrow$ is definable by a $\Sigma_2^-$-formula $\varphi$ given by \Cref{lem:sig2-}.
    Therefore, $\neg \varphi$ is a formula from $\Pi_2^+$ describing $L^\dualmon$.
\end{proof}

Finally, we can prove \Cref{thm:sigpimon}:

\begin{proof}
    Thanks to \cite{OneQuantifierAlternation}, it is straightforward that any language from $\Sigma_2^+ \cap \Pi_2^+$ is monotone and $\FOtw$-definable.

    Let $L$ be a monotone $\FOtw$-definable language.

    In particular, $L$ belongs to $\Sigma_2$ and is monotone.
    Thus, by \Cref{lem:sig2+}, $L$ belongs to $\Sigma_2^+$.
    Similarly, $L$ belongs to $\Pi_2^+$ by \Cref{pi2+}.
\end{proof}

This last result shows how close to capture monotone $\FOtw$-definable languages $\FOtwp$ is.
However, it does not seem easy to lift the equivalence $\Sigma_2 \cap \Pi_2 = \FOtw$ to their positive fragments as we did for the other classical equivalences in \Cref{sec:equiv}.
Indeed, the proof from \cite{OneQuantifierAlternation} relies itself on the proof of \cite{PolynomialClosure} which is mostly semantic while we are dealing with syntactic equivalences.

This immediately implies that a counter-example separating $\FO$-monotone from $\FOp$ cannot be in $\FOtw[\bin_<]$ as stated in the following corollary:

\begin{corollary}\label{FO2clot}
    Any monotone language described by an $\FOtw[\bin_<]$ formula is also described by an $\FOp$ formula.
\end{corollary}





    




If the monotone closure $L^{\uparrow}$ of a language $L$ described by a formula of $\FOtw[\bin_<]$ is in $\FOp$, nothing says on the other hand that $L^{\uparrow}$ is described by a formula of $\FOtw[\bin_<]$, or even of $\FOtw[\bin_0]$ as the counterexample $L=a^*bc^*de^*$ shows.
The monotone closure $L^{\uparrow}$ cannot be defined by an $\FOtw[\bin_0]$ formula.
This can be checked using for instance Charles Paperman's online software: \url{https://paperman.name/semigroup/}.
Notice that the software uses the following standard denominations: \textbf{DA} corresponds to $\FOtw[\bin_<]$, and \textbf{LDA} to $\FOtw[\bin_0]$.






We give the following conjecture, where $\FOtw$ can stand either for $\FOtw[\bin_<]$ or for $\FOtw[\bin_0]$

\begin{conjecture} \label{conj}~
\begin{itemize}
    \item A monotone language is definable in $\FOtw$ if and only if it is definable in $\FOtwp$.
    \item It is decidable whether a given regular language is definable in $\FOtwp$
\end{itemize}

\end{conjecture}

Since we can decide whether a language is definable in $\FOtw$ and whether it is monotone, the first item implies the second one.

\bibliography{ref}

\begin{thebibliography}{10}

\bibitem{AjtaiGurevich87}
Miklos Ajtai and Yuri Gurevich.
\newblock Monotone versus positive.
\newblock {\em J. ACM}, 34(4):1004–1015, October 1987.

\bibitem{muMALL}
David Baelde and Dale Miller.
\newblock Least and greatest fixed points in linear logic.
\newblock In Nachum Dershowitz and Andrei Voronkov, editors, {\em Logic for
  Programming, Artificial Intelligence, and Reasoning}, pages 92--106, Berlin,
  Heidelberg, 2007. Springer Berlin Heidelberg.

\bibitem{mucalc}
Julian Bradfield and Igor Walukiewicz.
\newblock The mu-calculus and model checking.
\newblock In {\em Handbook of Model Checking}, pages 871--919. Springer, 2018.

\bibitem{CostFun}
Thomas Colcombet.
\newblock {Regular Cost Functions, Part I: Logic and Algebra over Words}.
\newblock {\em {Logical Methods in Computer Science}}, {Volume 9, Issue 3},
  August 2013.
\newblock URL: \url{https://lmcs.episciences.org/1221}, \href
  {https://doi.org/10.2168/LMCS-9(3:3)2013}
  {\path{doi:10.2168/LMCS-9(3:3)2013}}.

\bibitem{UTL}
Kousha Etessami, Moshe~Y. Vardi, and Thomas Wilke.
\newblock First-order logic with two variables and unary temporal logic.
\newblock {\em BRICS Report Series}, 4(5), Jan. 1997.
\newblock URL: \url{https://tidsskrift.dk/brics/article/view/18784}, \href
  {https://doi.org/10.7146/brics.v4i5.18784}
  {\path{doi:10.7146/brics.v4i5.18784}}.

\bibitem{tempLogicBook}
Dov~M. Gabbay, Ian Hodkinson, and Mark Reynolds.
\newblock {\em Temporal logic (vol. 1): mathematical foundations and
  computational aspects}.
\newblock Oxford University Press, Inc., USA, 1994.

\bibitem{GrigniSipser92}
Michelangelo Grigni and Michael Sipser.
\newblock Monotone complexity.
\newblock In {\em Poceedings of the London Mathematical Society Symposium on
  Boolean Function Complexity}, page 57–75, USA, 1992. Cambridge University
  Press.

\bibitem{sep}
Ian~M. Hodkinson and Mark Reynolds.
\newblock Separation - past, present, and future.
\newblock In Sergei~N. Art{\"{e}}mov, Howard Barringer, Artur~S. d'Avila
  Garcez, Lu{\'{\i}}s~C. Lamb, and John Woods, editors, {\em We Will Show Them!
  Essays in Honour of Dov Gabbay, Volume Two}, pages 117--142. College
  Publications, 2005.

\bibitem{Immerman}
Neil Immerman.
\newblock Relational queries computable in polynomial time.
\newblock {\em Information and Control}, 68(1):86--104, 1986.
\newblock \href {https://doi.org/https://doi.org/10.1016/S0019-9958(86)80029-8}
  {\path{doi:https://doi.org/10.1016/S0019-9958(86)80029-8}}.

\bibitem{Kamp}
Hans Kamp.
\newblock {\em Tense Logic and the Theory of Linear Order}.
\newblock PhD thesis, University of California Los Angeles, 1968.
\newblock Published as Johan Anthony Willem Kamp.
\newblock URL:
  \url{http://www.ims.uni-stuttgart.de/archiv/kamp/files/1968.kamp.thesis.pdf}.

\bibitem{between}
Andreas Krebs, Kamal Lodaya, Paritosh~K. Pandya, and Howard Straubing.
\newblock Two-variable logics with some betweenness relations: Expressiveness,
  satisfiability and membership.
\newblock {\em Log. Methods Comput. Sci.}, 16(3), 2020.
\newblock URL: \url{https://lmcs.episciences.org/6765}.

\bibitem{PFO}
Denis Kuperberg.
\newblock Positive first-order logic on words and graphs.
\newblock {\em Log. Methods Comput. Sci.}, 19(3), 2023.
\newblock URL: \url{https://doi.org/10.46298/lmcs-19(3:7)2023}.

\bibitem{FOtoLTL}
Denis Kuperberg and Michael Vanden~Boom.
\newblock On the expressive power of cost logics over infinite words.
\newblock In Artur Czumaj, Kurt Mehlhorn, Andrew Pitts, and Roger Wattenhofer,
  editors, {\em Automata, Languages, and Programming}, pages 287--298, Berlin,
  Heidelberg, 2012. Springer Berlin Heidelberg.

\bibitem{Lyndon59}
Roger~C. Lyndon.
\newblock Properties preserved under homomorphism.
\newblock {\em Pacific J. Math.}, 9(1):143--154, 1959.
\newblock URL: \url{https://projecteuclid.org:443/euclid.pjm/1103039459}.

\bibitem{OliveiraR21}
Daniel Oliveira and Jo{\~{a}}o Rasga.
\newblock Revisiting separation: Algorithms and complexity.
\newblock {\em Log. J. {IGPL}}, 29(3):251--302, 2021.
\newblock URL: \url{https://doi.org/10.1093/jigpal/jzz081}, \href
  {https://doi.org/10.1093/JIGPAL/JZZ081} {\path{doi:10.1093/JIGPAL/JZZ081}}.

\bibitem{PolynomialClosure}
Jean-{\'E}ric Pin and Pascal Weil.
\newblock Polynomial closure and unambiguous product.
\newblock {\em Theory of Computing Systems}, 30:383--422, 1995.
\newblock URL: \url{https://api.semanticscholar.org/CorpusID:850708}.

\bibitem{Rossman08}
Benjamin Rossman.
\newblock Homomorphism preservation theorems.
\newblock {\em J. ACM}, 55, 07 2008.

\bibitem{Stol95}
Alexei~P. Stolboushkin.
\newblock Finitely monotone properties.
\newblock In {\em LICS, San Diego, California, USA, June 26-29, 1995}, pages
  324--330. {IEEE} Computer Society, 1995.

\bibitem{OneQuantifierAlternation}
Denis Th\'{e}rien and Thomas Wilke.
\newblock Over words, two variables are as powerful as one quantifier
  alternation.
\newblock In {\em Proceedings of the Thirtieth Annual ACM Symposium on Theory
  of Computing}, STOC '98, page 234–240, New York, NY, USA, 1998. Association
  for Computing Machinery.
\newblock \href {https://doi.org/10.1145/276698.276749}
  {\path{doi:10.1145/276698.276749}}.

\bibitem{Vardi}
Moshe~Y. Vardi.
\newblock The complexity of relational query languages (extended abstract).
\newblock In {\em Proceedings of the Fourteenth Annual ACM Symposium on Theory
  of Computing}, STOC '82, page 137–146, New York, NY, USA, 1982. Association
  for Computing Machinery.
\newblock \href {https://doi.org/10.1145/800070.802186}
  {\path{doi:10.1145/800070.802186}}.

\end{thebibliography}

\appendix

\section{Proof of Lemma \ref{LTLtoFO}}\label{app:LTLtoFO}

\begin{proof}

    
    Let us show the lemma by induction on the $\LTLp$ formula.
    We inductively construct for any formula $\varphi$ of $\LTLp$, a formula $\varphi^\bigstar(x)$ of $\FOthp$ with one free variable that describes the same language. This just amounts to remove the negation case in the classical proof, no additional difficulty here.

    \begin{itemize}
        \item $\toFO\bot = \bot$,
        \item $\toFO\top = \top$,
        \item $\toFO a = a(x)$,
        \item $\toFO{(\varphi \land \psi)}(x) = \varphi^\bigstar(x) \land \psi^\bigstar(x)$,
        \item $\toFO{(\varphi \lor \psi)}(x) = \varphi^\bigstar(x) \lor \psi^\bigstar(x)$,
        \item $\toFO{(\X \varphi)}(x) = \exists y, \succ(x,y) \land \varphi^\bigstar(y)$,
        \item $\toFO{(\varphi \U \psi)}(x) = \exists y, x \leq y \land \toFO{\psi}(y) \land \forall z, (z < x \lor y \leq z \lor \toFO\varphi(z))$,
        \item $\toFO{(\psi \R \varphi)}(x) = \toFO{(\varphi \U \psi )}(x)\vee(\forall y,y<x\vee \toFO{\varphi}(y))$.
    \end{itemize}

    The translation of a formula $\varphi$ of $\LTLp$ into a closed formula of $\FOthp$ is therefore $\exists x, x=0 \land \toFO\varphi(x)$, where $x=0$ is short for $\forall y, y \geq x$.

    This construction makes it possible to reuse the variables introduced. This is why we can translate the formulas of $\LTLp$ into $\FOthp$.
\end{proof}

\section{Monoids}

\subsection{Algebraic definitions} \label{ap:defmon}

\begin{definition}
    A semigroup is a pair $(\mathbf{S}, \op)$ where $\op$ is an associative internal composition law on the non-empty set $\mathbf{S}$.
\end{definition}

\begin{remark}
    We allow ourselves the abuse of language which consists in speaking of the semigroup $\mathbf{S}$ instead of the semigroup $(\mathbf{S}, \op)$.
\end{remark}

\begin{definition}
    A monoid is a pair $(\M, \op)$ which is a semigroup, and which has a neutral element noted $1_{\M}$ (or simply $1$ when there is no ambiguity), i.e. which verifies:
    $$
    \forall m \in \M, 1 \op m = m \op 1 = m.
    $$
\end{definition}

\begin{definition}
    Let $(\M, \op)$ and $(\M', \circ)$ be two monoids. An application $h$ defined from $\M$ into $\M'$ is a morphism of monoids if:
    $$
    \forall (m_1,m_2) \in \M^2, h(m_1 \op m_2) = h(m_1) \circ h(m_2),
    $$
    
    and
    $$
    h(1_{\M}) = 1_{\M'}.
    $$

    Similarly, if $\M$ and $\M'$ are just semigroups, $h$ is a morphism if it preserves the semigroup structure.
    
\end{definition}

\begin{definition}
    Let $(\M,\op)$ be a monoid, and $\leq$ an order on $\M$. We say that $\leq$ is compatible with $\op$ if:

    $$
    \forall (m,m',n,n') \in \M^4, m \leq n \land m' \leq n' \implies m \op m' \leq n \op n'.
    $$
\end{definition}

\begin{definition}
    Let $L$ be a language and $(\M, \op)$ a finite monoid. We say that $\M$ recognises $L$ if there exists a monoid morphism $h$ from $(\alph^*,.)$ into $(\M, \op)$ such that $L = h^{-1}(h(L))$.
\end{definition}

\begin{definition}
    Let $L$ be a regular language, and $u,v\in\alph^*$ be any two words. We define the equivalence relation of indistinguishability denoted $\sim_L$ on $\alph^*$. We write $u \sim_L v$ if:
    $$
    \forall (x,y) \in \alph^* \times \alph^*, xuy \in L \iff xvy \in L.
    $$

    Similarly, we write $u \leq_L v$ if:
    $$
    \forall (x,y) \in \alph^* \times \alph^*, xuy \in L \implies xvy \in L.
    $$

    The $\leq_L$ preorder is called the $L$ syntactic preorder.
\end{definition}

\begin{definition}
    Let $L$ be a regular language. We define the syntactic monoid of $L$ as $\M_L = L/\sim_L$.
\end{definition}

\begin{remark}
    This is effectively a monoid, since $\sim_L$ is compatible with left and right concatenation. Moreover, the syntactic monoid recognises $L$ through canonical projection. Moreover, we can see that the order $\leq_L$ naturally extends to an order compatible with the product on the syntactic monoid. We will use the same notation to designate both the pre-order $\leq_L$ and the order induced by $\leq_L$ on $\M_L$, which we will call syntactic order.
\end{remark}

\subsection{\texorpdfstring{Proof of \Cref{ordrmono}}{Proof of Lemma 29}} \label{ap:proofordrmono}

\begin{proof}
    The right-to-left direction follows from the definition of monotone languages. Indeed, suppose we have a language $L$ and an order $\leqml$ on its syntactic monoid that verifies the assumptions. Let $u$ be a word in $L$, and $v\geqa u$. By hypothesis, we have $h(v)\geqml h(u)$. Again by hypothesis, since $h(u)\in h(L)$, we also have $h(v)\in h(L)$, so $v$ belongs to $L$. We can conclude that $L$ is monotone.

    Conversely, let us consider a regular language $L$, and note $h$ its canonical projection onto its syntactic monoid. Let $\to$ be the binary relation induced by $\leqa$ on $\M_L$, i.e. such that $m \to n$ if there are words $u$ and $v$ such that $m=h(u)$, $n=h(v)$ and $u \leqa v$. The transitive closure of $\to$, denoted $\to^*$, is then an order relation.
 
    First of all, it is clearly reflexive and transitive.

    Then, to show antisymmetry, it is sufficient to show that $\to^*$ is included in $\leq_L$.

    Let $m$ and $n$ be two elements of $\M_L$ such that $m \to^* n$. By definition, there are $m_1$, $m_2$, ..., $m_p$ $p$ elements of $\M_L$ such that $m \to m_1 \to m_2 \to ... \to m_p \to n$,
    where $p$ is a natural number.
    We then have $u_0$, $u_1$, $u_1'$, $u_2$, $u_2'$, ..., $u_p$, $u_p'$, and $u_{p+1}$ such that $m = h(u_0)$, $m_1 = h(u_1) = h(u_1')$, $m_2 = h(u_2) = h(u_2')$, ..., $m_p = h(u_p) = h(u_p')$ and $n = h(u_{p+1})$ and $u_0 \leqa u_1$, $u_1' \leqa u_2$, $u_2' \leqa u_3$, ..., $u_p' \leqa u_{p+1}$.

    Now let $x$ and $y$ be two words (they constitute a context). By monotonicity of $L$, if $xu_0y$ belongs to $L$, then $xu_1y$ belongs to $L$. Then, since $h(u_1) = h(u_1')$, if $xu_1y$ belongs to $L$, then so does $xu_1'y$. We immediately deduce that if $xu_0y$ belongs to $L$, then so does $xu_{p+1}y$. This proves that $\to^*$ is included in $\leq_L$.

    So $\to^*$ is an order, which we note $\leqml$.

    Let us check its compatibility with the operation $\op$ of the monoid.
    Let $m$, $m'$, $n$ and $n'$ be elements of $\M_L$ such that $m \leqml n$ and $m' \leqml n'$.

    First, let us assume $m \to n$ and $m' \to n'$. We then have $u$, $u'$, $v$ and $v'$ representing $m$, $m'$, $n$ and $n'$ respectively, such that $u \leqa v$ and $u' \leqa v'$. So we have $uv \leqa u'v'$ and thus, $mn \leqml m'n'$. Now, if we only have $m \to^* n$ and $m' \to^* n'$, then we have finite sequences $(m_i)_{i=1}^p$ and $(m_i')_{i=1}^p$, which we can assume to be of the same length $p$ by reflexivity of $\to$, such that $m \to m_1 \to ... \to m_p \to n$ and $m' \to m_1' \to ... \to m_p' \to n'$. So we have $m \op m' \leqml m_1 \op m_1'$, but also $m_1 \op m_1' \leqml m_2 \op m_2'$, ..., $m_p \op m_p' \leqml n \op n'$. We then obtain the inequality $mn \leqml m'n'$ by transitivity.


    Finally, it is clear that if $u\leqa v$ then $h(u)\leqml h(v)$.


    The relationship $\leqml$ therefore satisfies the constraints imposed.
    
\end{proof}

\subsection{\texorpdfstring{Proof of \Cref{complexity}}{Proof of Proposition 32}} \label{ap:proofmon}

\begin{proof}
    First, in the algorithm from the \Cref{prop:algo}, at any given time, we only need to code two letters from $\alph$ and two elements from the monoid $\M$. 
    So we can code $S$ and $S'$ with $|\Sigma|$ bits and increment them through the loop in order to go through the whole alphabet.
    For example, if $\Sigma = \{a,b,c\}$ then $a$ is coded by $001$, $\{a,b\}$ by $010$ and so on.
    In the same way, we only need $2\lceil\log_2(\M)\rceil$ bits to code $(m,n)$.
    Using lookup tables for applying the function $h$, the product $\cdot$, and testing membership in $F$, all operations can be done in $\L$.
    Thus, the algorithm from the \Cref{prop:algo} is in $\L$.

    To decide whether a DFA $\aut$ describes a monotone language, we can compute the NFA $\aut^{\uparrow}$ by adding to each transition $(q_0,a,q_1)$ of $\aut$ any transition $(q_0,b,q_1)$ with $b$ greater than $a$.
    Thus, $\aut^{\uparrow}$ describes the monotone closure of the language recognised by $\aut$.
    Then, $\aut$ recognises a monotone language if and only if there is not path from an initial to a final state in the product automaton $\overline{\aut} \times \aut^{\uparrow}$, where $\overline{\aut}$ is the complement of $\aut$, obtained by simply switching accepting and non-accepting states. 
    As NFA emptiness is in $\NL$, DFA monotonicity is in $\NL$ as well.

    Now, let us suppose we have an algorithm which takes a DFA as input and returns whether it recognises a monotone language.
    Notice that the DFA emptiness problem is still $\NL-\comp$ when restricted to automata not accepting the empty word $\varepsilon$. We will use this variant to perform a reduction to DFA monotonicity.
    Suppose we are given a DFA $\aut$ on an alphabet $A$ which does not accept $\varepsilon$.
    We build an automaton $\aut'$ on $A \cup \{\top\}$ by adding the letter $\top$ to $A$ in $\aut$, but without any $\top$-labelled transition.
    Now, let us equip $A \cup \{\top\}$ with an order $\leq$ such that $a \leq \top$ for any letter $a$ of $A$.
    Then the new automaton $\aut'$ recognises a monotone language if and only if $\aut$ recognises the empty language.
    Indeed, suppose we have a word $u$ of length $n$ accepted by $\aut$.
    Then, $\aut'$ would accept $u$ but not $\top^n$ which is bigger than $u$.
    Reciprocally, if $\aut$ recognises the empty language then so does $\aut'$ and the empty language is a monotone language.
    Thus, the monotonicity problem is $\NL-\comp$ when the input is a DFA.

\end{proof}


\section{\texorpdfstring{An $\FOtw[\bin_0 \cup \be]$-formula for the counter-example}{An FO2[<,S,,be]-formula for the counter-example}} \label{ap:formula}

Let us give a formula for the counter-example from \cref{FO2be}.

Let us notice that the successor predicate is definable in $\FOtw[\bin_< \cup \be]$, so results from \cite{between} about the fragment $\FOtw[<,\be]$ apply to $\FOtw[\bin_0 \cup \be]$ as well.


So it is easy to describe $A^*(
        \top \cup
        \binom{a}{b}^2 \cup \binom{b}{c}^2 \cup \binom{c}{a}^2 \cup
        \binom{a}{b}\binom{c}{a} \cup
        \binom{b}{c}\binom{a}{b} \cup
        \binom{c}{a}\binom{b}{c}
    )A^*$
and to state that factors of length $3$ are in $(abc)^\uparrow$.

Now, for any atomic predicates $s$ and $t$ (i.e. $s,t\in\{a,b,c\}$), let us pose:

$$
\varphi_{s,t} = \forall x, \forall y,
\left(
    s(x) \land t(y) \land x<y \land \bigwedge_{d \in \Sigma} \neg d(x,y)
\right)
\implies
\psi_{s,t}(x,y),
$$

where $\psi_{s,t}(x,y)$ is a formula stating that the two anchors are compatible, i.e. either they both use the ``upper component'' of all the double letters between them, or they both use the ``bottom component''.
Recall that $\bigwedge_{d \in \Sigma}\neg d(x,y)$ means that there is no singleton letter between $x$ and $y$.

For example, $\psi_{a,b}(x,y)$ is the disjunction of the following formulas:
$$
\def\arraystretch{1.5}
    \begin{array}{c}
\binom{b}{c}(x+1)  \wedge \binom{a}{b}(y-1)\\
 \binom{a}{b}(x+1)  \wedge  \binom{c}{a}(y-1)\\
x+1=y
\end{array}
$$

Indeed, the first case correspond to using the upper component of $\binom{b}{c}$ and $\binom{a}{b}$: anchor $a$ in position $x$ is followed by the upper $b$ in position $x+1$, which should be consistent with the upper $a$ in position $y-1$ followed by anchor $b$ in position $y$, the factor from $x+1$ to $y-1$ being of the form $(\binom{b}{c}\binom{c}{a} \binom{a}{b})^+$. 
Similarly, the second case corresponds to the bottom component.
The last case corresponds to anchors directly following each other, without an intermediary factor of double letters. This case appears only for $(s,t)\in\{(a,b),(b,c),(c,a)\}$

Now using the conjunction of all formulas $\varphi_{s,t}$ where $s$ and $t$ are atomic predicates $a,b,c$, we build a formula for the language of \cref{FO2be}.

\section{Games} \label{ap:games}

Erhenfeucht-Fraïssé games and their variants are traditionally used to prove negative expressivity results of $\FO$ fragments. This is why we were interested in Erhenfeucht-Fraïssé games matching fragments of $\FOp$.
Although we did not manage to use them in the present work, we include here a variant that could be suited for proving $\FOtwp$ inexpressibility results.

\begin{definition}
    We note $\EF_k^{n+}[\bin](u_0,u_1)$, the Ehrenfeucht-Fraïssé game associated with $\FO^{n+}[\bin]$ at $k$ turns on the pair of words $(u_0,u_1)$. When there is no ambiguity, we simply note $\EF_k^{n+}(u_0,u_1)$. In $\EF_k^{n+}(u_0,u_1)$, two players, Spoiler and Duplicator, play against each other on the word pair $(u_0,u_1)$ in a finite number $k$ of rounds. Spoiler and Duplicator will use tokens numbered $1$, $2$, ..., $n$ to play on the positions of the words $u_0$ and $u_1$.
    
    On each turn, Spoiler begins. He chooses $\delta$ from $\{0,1\}$ and $i$ from $[\![1,n]\!]$ and moves (or places, if it has not already been placed) the $i$ numbered token onto a position of the word $u_{\delta}$. Duplicator must then do the same on the word $u_{1-\delta}$ with the constraint of respecting binary predicates induced by the placement of the tokens, and only in one direction for unary predicates. More precisely, if $\val_0$ and $\val_1$ are the valuations that to each token (considered here as variables) associates the position where it is placed in $u_0$ and $u_1$ respectively, then
    \begin{itemize}
        \item for any binary predicate $\b(x,y)$, $(u_0,\val_0)\models\b(x,y)$ if and only if $(u_1,\val_1)\models\b(x,y)$,
        \item for any unary predicate $a(x)$ in $\Sigma$, if $(u_0,\val_0)\models a(x)$ then $(u_1,\val_1)\models a(x)$.
    \end{itemize}
    
    If Duplicator cannot meet the constraint, he loses and Spoiler wins.

    In particular, for any $i\in[\![1,n]\!]$, if the letter $s_0$ indicated by the token $i$ on the word $u_0$ is not included in the letter $s_1$ indicated by the token $i$ on the word $u_1$, then Spoiler wins.

    If after $k$ rounds, Spoiler has not won, then Duplicator is declared the winner.
\end{definition}

\begin{theorem} \label{EF2+}
    Let $L$ be a language and $n$ a natural number. The language $L$ is definable by a formula of $\FO^{n+}[\bin]$ if and only if there exists a natural number $k$ such that, for any pair of words $(u_0,u_1)$ where $u_0$ belongs to $L$ but $u_1$ does not, Spoiler has a winning strategy in $\EF_k^{n+}[\bin](u_0,u_1)$.
\end{theorem}

\begin{proof}
    We generalise the proof from \cite[Theorem 5.7]{PFO}, which treats the case of $\FOp$, using a classical construction for $\FO$ with a bounded number of variables.

    Let $n$ be a natural number.
    Let us introduce the concept of initial configuration. For two words $u_0$ and $u_1$ of lengths $l_0$ and $l_1$ respectively, and two functions of $0$, $1$, $2$, ..., or $n$ variables among $x_1$, ... $x_n$, $\val_0$ and $\val_1$ with values in $[\![0,l_0-1]\!]$ and $[\![0,l_1-1]\! ]$ respectively, the game $\EF_k^{n+}[\bin](u_0,u_1)$ has initial configuration $(\val_0,\val_1)$ if token $i$ is placed in position $\val_0(x_i)$ on word $u_0$, when $\val_0(x_i)$ is defined, for any integer $i$ from $[\! [1,n]\!]$, and similarly with $u_1$ for the valuation $\val_1$.

    We then claim that for any natural number $k$ and any formula $\varphi$ of $\FO^{n+}[\bin]$ (possibly with free variables) of quantification rank at most $k$, and for all models $(u_0,\val_0)$ and $(u_1,\val_1)$, Duplicator wins the game $\EF_k^{n+}[\bin](u_0,u_1)$ with initial configuration $(\val_0,\val_1)$, if and only if:
    $$
        u_0,\val_0 \models \varphi \implies u_1,\val_1 \models \varphi.
    $$

    Indeed, starting from the induction from the article \cite{PFO}, we have to adapt the base case to the set of binary predicates $\bin$ considered. The proof is then similar: each element of $\bin$ can impose a constraint in $\FO^{n+}[\bin]$ which is reflected in the constraint on the positions of the tokens. Then, in the induction, we need to modify the valuation update. Indeed, as the number of variables (and therefore of tokens) is limited to $n$, when a variable $x$ already in use is encountered, we do not need to add a variable to the valuation $\val$ constructed, but modify the value taken by $\val$ in $x$, to construct a new valuation $\val'$.

\end{proof}

\end{document}